\documentclass{article}

\usepackage{amsmath}
\usepackage{amsthm}
\usepackage{amssymb}
\usepackage{hyperref}
\usepackage{color}
\usepackage{graphicx}
\usepackage{enumerate}
\usepackage{bm}
\usepackage[margin = 1in]{geometry}
\usepackage{bbm}
\usepackage{verbatim}
\usepackage{caption}

\usepackage{algorithm}
\usepackage{algpseudocode}

\theoremstyle{definition} \newtheorem{definition}{Definition}[section]
\theoremstyle{plain} \newtheorem{theorem}[definition]{Theorem}
\theoremstyle{plain} 
\theoremstyle{plain} 
\theoremstyle{plain} 
\theoremstyle{plain} \newtheorem{lemma}[definition]{Lemma}
\theoremstyle{plain} \newtheorem{observation}[definition]{Observation}
\theoremstyle{plain} 
\theoremstyle{definition} 
\theoremstyle{plain} 
\theoremstyle{definition} 
\theoremstyle{remark} 

\newcommand{\eps}{\ensuremath{\varepsilon}}

\newcommand{\mcN}{\ensuremath{\mathcal{N}}}
\newcommand{\mcP}{\ensuremath{\mathcal{P}}}
\newcommand{\mfX}{\ensuremath{\mathfrak{X}}}
\newcommand{\broder}{\ensuremath{\mathfrak{X}_{\textrm{B}}}}
\newcommand{\jsv}{\ensuremath{\mathfrak{X}_{\textrm{JSV}}}}
\newcommand{\permanentalg}{\textsc{Permanent}}

\renewcommand{\P}{\mathcal{P}}
\newcommand{\N}{\mathcal{N}}
\newcommand{\fpras}{\mathsf{FPRAS}}
\newcommand{\poly}{{\sf poly}}
\newcommand{\Tmix}{T_{\mathrm{mix}}}
\newcommand{\dtv}{d_{\mathsf{TV}}}

\title{On Counting Perfect Matchings in General Graphs}

\author{Daniel
\v{S}tefankovi\v{c}\thanks{University of Rochester, USA. Email: \texttt{stefanko@cs.rochester.edu}}
\and
Eric Vigoda\thanks{Georgia Institute of Technology, USA. Email:
\texttt{\{vigoda,wilmesj\}@gatech.edu}}
\and 
John Wilmes$^\dag$}

\begin{document}
\maketitle
\begin{abstract}
Counting perfect matchings has played a central role in the theory of counting problems.
The permanent, corresponding to bipartite graphs, was shown to be \#P-complete to
compute exactly by Valiant (1979),
and a fully polynomial randomized approximation scheme (FPRAS) was presented
by Jerrum, Sinclair, and Vigoda (2004)
  using a Markov chain Monte Carlo (MCMC) approach. However, it has
    remained an open question whether there exists an FPRAS for counting
  perfect  matchings in general graphs. In fact, it was unresolved
      whether the same Markov chain defined by JSV is rapidly mixing in general.
    In this paper, we show that it is not. We prove torpid mixing for
any weighting scheme on hole patterns in the JSV chain.
    As a first step toward overcoming this obstacle, we
    introduce a new algorithm for counting matchings based on the
    Gallai--Edmonds decomposition of a graph, and give an FPRAS for counting
    matchings in graphs that are sufficiently close to bipartite. In
    particular, we obtain a fixed-parameter tractable algorithm for counting
    matchings in general graphs, parameterized by the greatest ``order'' of a
    factor-critical subgraph.
\end{abstract}

\section{Introduction}\label{sec:intro}

Counting perfect matchings is a fundamental problem in the area of counting/sampling problems.
For an undirected graph $G=(V,E)$, let $\P$ denote the set of perfect matchings of $G$.
Can we compute (or estimate) $|\P|$ in time polynomial in $n=|V|$?
For which classes of graphs?

A polynomial-time algorithm for the corresponding decision and optimization
problems of determining if a given graph contains a
perfect matching or finding a matching of maximum size was presented
by Edmonds~\cite{Edmonds}.
For the counting problem, a classical algorithm of Kasteleyn~\cite{Kasteleyn}
gives a polynomial-time algorithm for exactly computing $|\P|$ for planar graphs.

For bipartite graphs, computing $|\P|$ is equivalent to computing the permanent of $n\times n$ $(0,1)$-matrices.   Valiant~\cite{Valiant} proved that the $(0,1)$-Permanent is \#P-complete.
Subsequently attention turned to the Markov Chain Monte Carlo (MCMC) approach.
A Markov chain where the mixing time is polynomial in $n$ is said to be {\em rapidly mixing},
and one where the mixing time is exponential in $\Omega(n)$ is referred to as
{\em torpidly mixing}.  A rapidly mixing chain yields an $\fpras$ (fully polynomial-time
randomized approximation scheme) for the corresponding counting problem of
estimating $|\P|$~\cite{JVV}.

For dense graphs, defined as those with minimum degree $>n/2$, Jerrum and Sinclair~\cite{JS}
proved rapid mixing of a Markov chain defined by Broder~\cite{Broder}, which yielded an $\fpras$ for
estimating $|\P|$.   The Broder chain walks on the collection $\Omega=\P\cup\N$ of
perfect matchings $\P$
and near-perfect matchings $\N$; a near-perfect matching is a matching with
exactly 2 holes or unmatched vertices.
Jerrum and Sinclair~\cite{JS}, more generally, proved rapid mixing when
the number of perfect matchings is within a $\poly(n)$ factor of the
number of near-perfect matchings, i.e., $|\P|/|\N|\geq 1/\poly(n)$.
A simple example, referred to as a ``chain of boxes'' which is
illustrated in Figure~\ref{fig:chain-of-boxes},
shows that the Broder chain is torpidly mixing.   This example was a useful
testbed for catalyzing new approaches to solving the general permanent problem.

Jerrum, Sinclair and Vigoda~\cite{JSV} presented a new Markov chain on $\Omega=\P\cup\N$
with a non-trivial weighting scheme on the matchings based on the holes (unmatched vertices).
They proved rapid mixing for any bipartite graph with the requisite weights used in the
Markov chain, and they presented a polynomial-time algorithm to learn these weights.
This yielded an $\fpras$ for estimating $|\P|$ for all bipartite graphs.  That is the current
state of the art (at least for polynomial-time, or even sub-exponential-time algorithms).

Could the JSV-Markov chain be rapid mixing on non-bipartite graphs?  Previously
there was no example for which torpid mixing was established, it was simply the case
that the proof in~\cite{JSV} fails.  We present a relatively simple example where the
JSV-Markov chain fails for the weighting scheme considered in~\cite{JSV}.  More
generally, the JSV-chain is torpidly mixing on our class of examples for any weighting
scheme based on the hole patterns, see Theorem~\ref{thm:jsv-fails} in
Section~\ref{sec:JSV-fails} for a formal statement following the precise definition
of the JSV-chain.

A natural approach for non-bipartite graphs is to consider Markov chains
that exploit odd cycles or blossoms in the manner of Edmonds' algorithm.
We observe that a Markov chain which considers \emph{all} blossoms for its
transitions is intractable since sampling all blossoms is NP-hard,
see~Theorem~\ref{thm:all-blossoms-hard}. On the other hand, a chain restricted
to minimum blossoms is not powerful enough to overcome our torpid mixing
examples. See Section~\ref{sec:blossoms-fail} for a discussion.

Finally we utilize the Gallai--Edmonds graph decomposition into
factor-critical graphs~\cite{Edmonds,GallaiA,GallaiB,Schrijver} to present
new algorithmic insights that may overcome the obstacles in our classes of
counter-examples.  In Section~\ref{sec:EG}, we describe how the Gallai--Edmonds
decomposition can be used to efficiently estimate $|\P|$, the number of perfect
matchings, in graphs whose factor-critical subgraphs have bounded order
(Theorem~\ref{thm:fpt-alg}), as well as in the torpid mixing example
graphs (Theorem~\ref{thm:counterexample-alg}).

Although all graphs are explicitly defined in the text below, figures depicting
these graphs are deferred to the appendix,

%\section{Preliminaries}
\subsection{Markov Chains}

Consider an ergodic
Markov chain with transition matrix $P$ on a finite state space $\Omega$, and let
$\pi$ denote the unique stationary distribution. We will usually assume the
Markov chain is time reversible, i.e., that it satisfies the \textbf{detailed
balance condition} $\pi(x)P(x,y) = \pi(y)P(x,y)$ for all states $x, y \in
\Omega$.

For a pair of distributions $\mu$ and $\nu$ on $\Omega$ we denote
their total variation distance as $\dtv(\mu,\nu) = \frac{1}{2}\sum_{x\in\Omega}
|\mu(x)-\nu(x)|$.
The
standard notion of \textbf{mixing time}
$\Tmix$ is the number of steps from
the worst starting state $X_0=i$ to reach total variation distance $\leq 1/4$ of
the stationary distribution $\pi$, i.e., we write
$\Tmix = \max_{i\in\Omega} \min\{t: \dtv(P^t(i,\cdot),\pi)\leq 1/4\}$.

We use conductance to obtain lower bounds on the mixing time.
For a set $S\subset\Omega$ its \textbf{conductance} is defined as:
\[  \Phi(S) = \frac{\sum_{x\in S,y\notin S}\pi(x)P(x,y)}{\sum_{x\in S}\pi(x)}.
\]
Let $\Phi_* = \min_{S\subset\Omega:\pi(S)\leq 1/2} \Phi(S)$.
Then (see, e.g., \cite{Sinclair,LWP})
\begin{equation}
\label{eq:conductance}
 \Tmix \geq \frac{1}{4\Phi_*}.
\end{equation}

\subsection{Factor-Critical Graphs}

A graph $G = (V,E)$ is \textbf{factor-critical} if for every vertex $v \in V$,
the graph induced on $V\setminus \{v\}$ has a perfect matching. (In particular,
$|V|$ is odd.)

Factor-critical graphs are characterized by their ``ear'' structure. The
\textbf{quotient} $G/H$ of a graph $G$ by a (not necessarily induced) subgraph
$H$ is derived from $G$ by deleting all edges in $H$ and contracting all
vertices in $H$ to a single vertex $v_H$ (possibly creating loops or
multi-edges).  An \textbf{ear} of $G$ relative a subgraph $H$ of $G$ is simply a cycle in
$G/H$ containing the vertex $v_{H}$.

\begin{theorem}[Lov\'asz~\cite{Lovasz72}]
    A graph $G$ is factor-critical if and only if there is a decomposition $G =
    C_0 \cup \cdots \cup C_r$ such that $C_0$ is a single vertex, and $C_i$ is
    an odd-length ear in $G$ relative to $\bigcup_{j < i} C_j$, for all $0 <
    i \le r$.

    Furthermore, if $G$ is factor critical, there exists such a decomposition
    for every choice of vertex $C_0$, and the \emph{order} $r$ of the
    decomposition is independent of all choices.
\end{theorem}

Since the number of ears in the ear decomposition of a factor-critical graph
depends only on the graph, and not on the choice made in the decomposition, we
say the \textbf{order} of the factor-critical graph $G$ is the number $r$ of
ears in any ear decomposition of $G$.

Factor-critical graphs play a central role in the Gallai--Edmonds structure
theorem for graphs. We state an abridged version of the theorem below.

Given a graph $G$, let $D(G)$ be the set of vertices that remain unmatched in
at least one maximum matching of $G$. Let $A(G)$ be the set of vertices not in
$D(G)$ but adjacent to at least one vertex of $D(G)$. And let $C(G)$ denote the
remaining vertices of $G$.

\begin{theorem}[Gallai--Edmonds Structure Theorem]
    The connected components of $D(G)$ are factor-critical. Furthermore, every
    maximum matching of $G$ induces a perfect matching on $C(G)$, a near-perfect
    matching on each connected component of $D(G)$, and matches all vertices in
    $A(G)$ with vertices from distinct connected components of $D(G)$.
\end{theorem}

\section{The Jerrum--Sinclair--Vigoda Chain}
\label{sec:JSV-fails}

We recall the definition of the original Markov chain proposed by
Broder~\cite{Broder}. The
state space of the chain is $\Omega = \mcP \cup \bigcup_{u,v} \mcN(u, v)$
where $\mcP$ is the collection of perfect matchings and
$\mcN(u,v)$ are near-perfect matchings with holes at $u$ and $v$
(i.e., vertices $u$ and $v$ are the only unmatched vertices).
The transition rule for a matching $M \in \Omega$ is as follows:
\begin{enumerate}
    \item If $M \in \mcP$, randomly choose an edge $e \in M$ and transition to
        $M\setminus \{e\}$.
    \item If $M \in \mcN(u,v)$, randomly choose a vertex $x \in V$. If $x \in
        \{u,v\}$ and $u$ is adjacent to $v$, transition to $M \cup \{(u,v)\}$. Otherwise,
        let $y \in V$ be the vertex matched with $x$ in $M$, and randomly choose $w
        \in \{u,v\}$. If $x$ is adjacent to $w$, transition to the matching $M \cup
        \{(x,w)\}\setminus\{(x,y)\}$.
\end{enumerate}

The chain $\broder$ is symmetric, so its stationary distribution is uniform. In
particular, when $|\mcP|/|\Omega|$ is at least inverse-polynomial in $n$,
we can efficiently generate uniform samples from $\mcP$ via rejection sampling,
given access to samples from the stationary distribution of $\broder$.

In order to sample perfect matchings even when $|\Omega|/|\mcP|$ is
exponentially large, Jerrum, Sinclair, and Vigoda~\cite{JSV} introduce a new chain $\jsv$
that changes the stationary distribution of $\broder$ by means of a Metropolis
filter. The new stationary distribution is uniform across \emph{hole patterns},
and then uniform within each hole pattern, i.e., for every $M \in \Omega$, the
stationary probability of $M$ is proportional to $1/|\mcN(u,v)|$ if $M \in
\mcN(u,v)$, and proportional to $1/|\mcP|$ if $M \in \mcP$.

We define $\jsv$ in greater detail. For $M \in \Omega$, define the weight
function
\begin{equation}\label{eq:std-wt}
    w(M) = \left\{\begin{matrix} \frac{1}{|\mcP|} & \textrm{if $M \in \mcP$} \\
                                \frac{1}{|\mcN(u,v)|} & \textrm{if $M \in
    \mcN(u,v)$} \end{matrix}\right.
\end{equation}
\begin{definition}\label{def:jsv-chain}
    The chain $\jsv$ has the same state space as $\broder$. The transition rule
    for a matching $M \in \Omega$ is as follows:
    \begin{enumerate}
        \item First, choose a matching $M' \in \Omega$ to which
            $M$ may transition, according to the transition rule for $\broder$
        \item With probability $\min\{1, w(M')/w(M)\}$, transition to $M'$.
            Otherwise, stay at $M$.
    \end{enumerate}
\end{definition}

In their paper, Jerrum, Sinclair, and Vigoda~\cite{JSV} in fact analyze a more general
version of the chain $\jsv$ that allows for arbitrary edge weights in the graph.
They show that the chain is rapidly mixing
for bipartite graphs $G$. (They also study the separate problem of estimating
the weight function $w$, and give a ``simulating annealing'' algorithm that
allows the weight function $w$ to be estimated by gradually adjusting edge
weights to obtain an arbitrary bipartite graph $G$ from the complete bipartite
graph.) Their analysis of the mixing time uses a canonical paths argument that
crucially relies on the bipartite structure of the graph. However, it remained
an open question whether a different analysis of the same chain $\jsv$, perhaps
using different canonical paths, might generalize to non-bipartite graphs. We
rule out this approach.

In fact, we rule out a more general family of Markov chains for sampling
perfect matchings. We say a Markov chain is ``of $\jsv$ type'' if it has
the same state space as $\jsv$, with transitions as defined in
Definition~\ref{def:jsv-chain}, for \emph{some} weight function $w(M)$ (not
necessarily the same as in Eq.~\eqref{eq:std-wt}) depending only the hole
pattern of the matching $M$.

\begin{theorem}\label{thm:jsv-fails}
    There exists a graph $G$ on $n$ vertices such that for any Markov chain
    $\mfX$ of $\jsv$ type on $G$, either the stationary probability of $\mcP$
    is at most $\exp(-\Omega(n))$, or the mixing time of $\mfX$ is at least
    $\exp(\Omega(n))$.
\end{theorem}

The graph $G$ of Theorem~\ref{thm:jsv-fails} is constructed from several copies
of a smaller gadget $H$, which we now define.

\begin{definition}\label{def:chain-of-boxes}
    The \textbf{chain of boxes gadget} $B_k$ of length $k$ is the graph on $4k$
    vertices depicted in Figure~\ref{fig:chain-of-boxes}. To construct $B_k$,
    we start with a path $P_{2k-1} = v_0, v_1, \ldots, v_{2k}$ of length
    $2k-1$. Then, for every even edge $\{v_{2i}, v_{2i+1}\}$ on the path, we add two
    additional vertices $a_i, b_i$, along with edges to form a path $v_{2i}, a_i, b_i,
    v_{2i+1}$ of length $3$.
\end{definition}

\begin{figure}[h]
    \begin{centering}
        \includegraphics[width=0.5\textwidth]{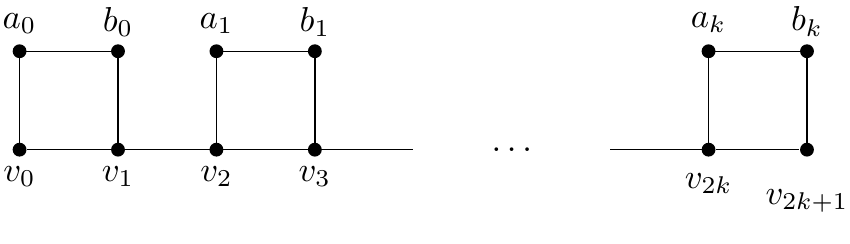}

    \end{centering}
    \caption{The ``chain of boxes'' gadget $B_k$, which has $2^k$ perfect
    matchings, but only a single matching in $\mcN(v_0, v_{2k+1})$.}
    \label{fig:chain-of-boxes}
\end{figure}

\begin{observation}\label{obs:chain-of-boxes}
    The chain of boxes gadget $B_k$ has $2^k$ perfect matchings, but only one matching
    in $\mcN(v_0, v_{2k+1})$.
\end{observation}

\begin{definition}\label{def:torpid-gadget}
    The \textbf{torpid mixing gadget} $H_k$ is the graph depicted in
    Figure~\ref{fig:gadget}. To construct $H$, first take a $C_{12}$ and label
    two antipodal vertices as $a$ and $b$. Add an edge between $a$ and $b$, and
    label the two vertices farthest from $a$ and $b$ as $u$ and $v$. Label the
    neighbor of $u$ closest to $a$ as $w_1$, and the other neighbor of $u$ as
    $w_2$. Label the neighbor of $v$ closest to $a$ as $z_1$ and the other
    neighbor of $v$ as $z_2$. Finally, add four chain-of-boxes gadgets $B_k$,
    identifying the vertices $v_0$ and $v_{2k}$ of the gadgets with $w_1$ and
    $a$, with $a$ and $z_1$, with $w_2$ and $b$, and with $b$ and $z_2$,
    respectively.
\end{definition}

Note that in Figures~\ref{fig:gadget} and~\ref{fig:gadget-alt}, one ``box'' from each
copy of $B_k$ in the torpid mixing gadget is left undrawn, for visual clarity.

\begin{figure}[h]
    \begin{centering}
    \includegraphics[width=0.7\textwidth]{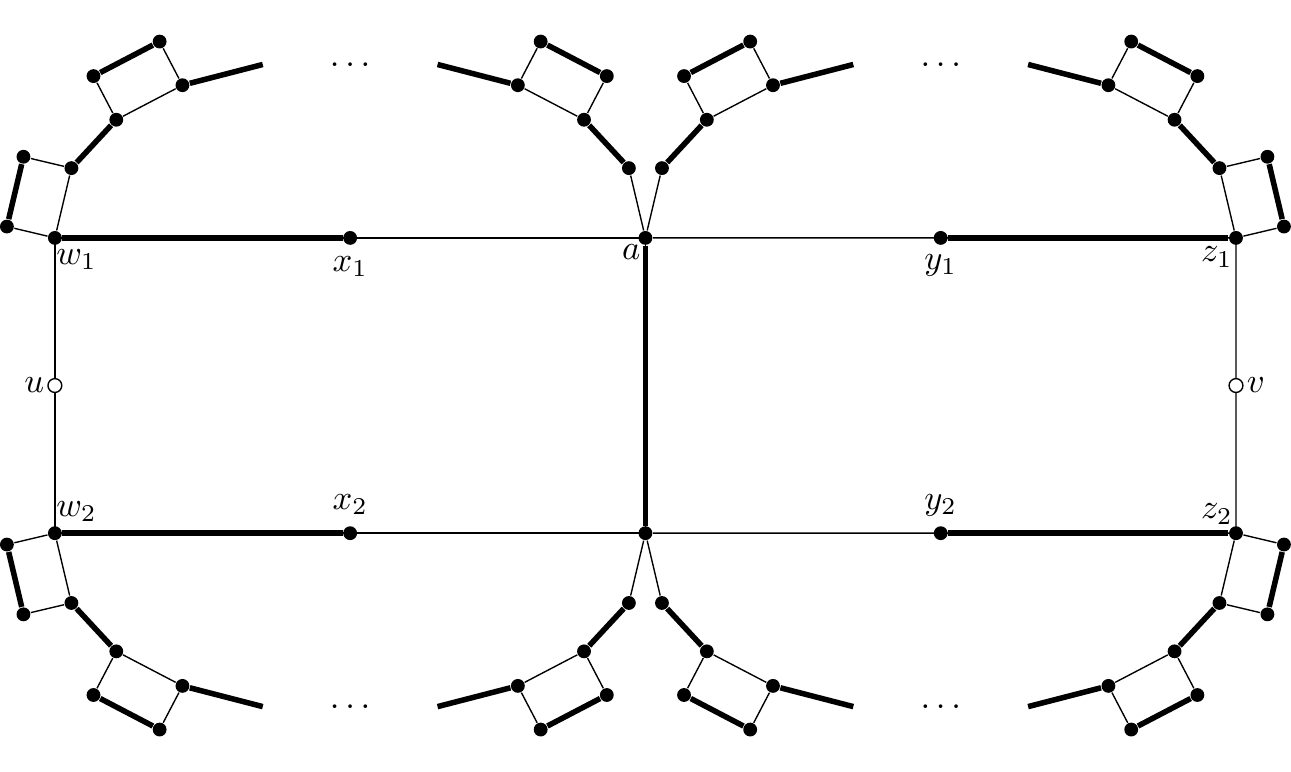}

    \end{centering}
    \caption{The torpid mixing gadget $H_k$. The unique matching $M
    \in \mcN(u,v)$ is depicted with thick edges.}
    \label{fig:gadget}
\end{figure}

\begin{figure}[h]
    \begin{centering}
    \includegraphics[width=0.7\textwidth]{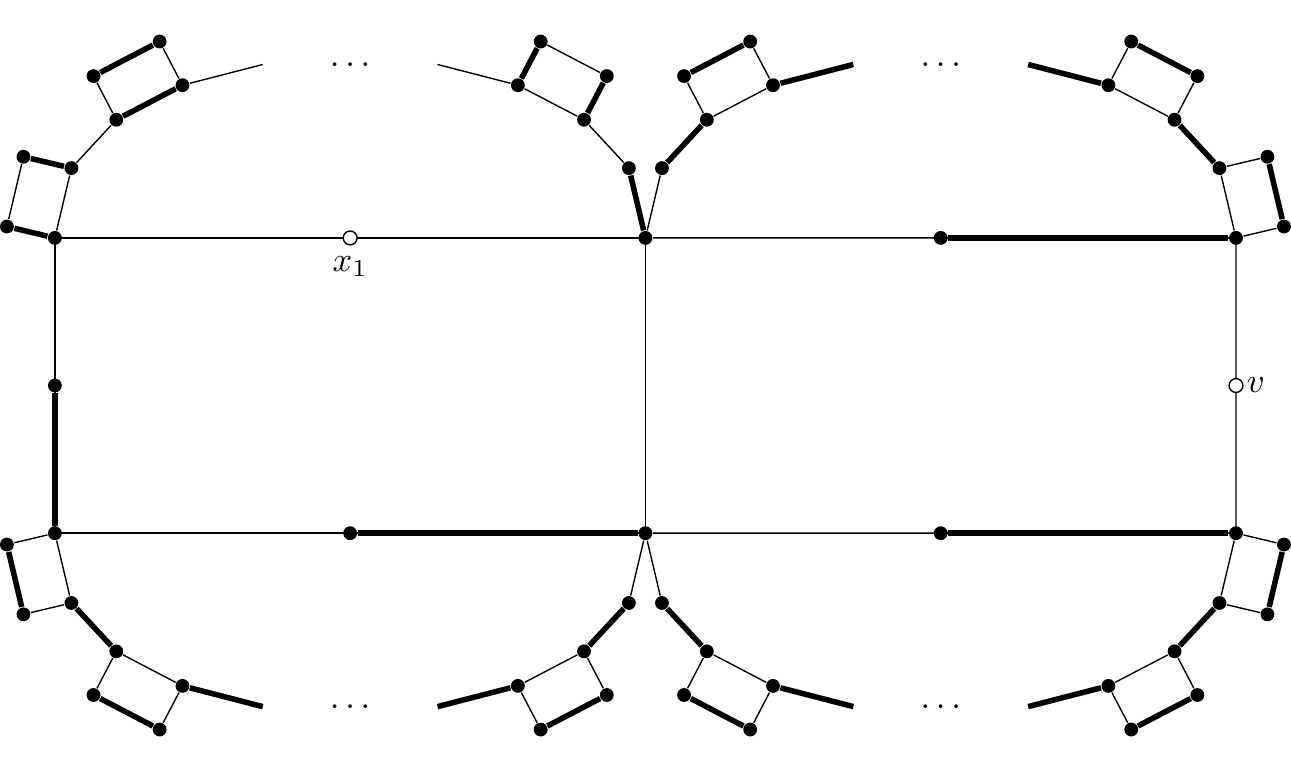}

    \end{centering}
    \caption{A matching $M' \in \mcN(x_1, v)$.  There are exponentially many
    matchings with the same hole pattern, obtained by alternating the
    $4$-cycles above $x_1$.}
    \label{fig:gadget-alt}
\end{figure}

\begin{lemma}\label{lem:torpid-gadget}
    The torpid mixing gadget $H = H_k$ has $16k + 4$ vertices and exactly $2$
    perfect matchings. Furthermore, $|\mcN_H(u,v)| = 1$ and $\mcN_H(x_1, v) \ge
    2^k$.
\end{lemma}
\begin{proof}
    A matching $M \in \mcN_H(u,v)$ is depicted in Figure~\ref{fig:gadget}. We
    argue that $M$ is the only matching in $\mcN(u,v)$. First note that $x_1$
    must be matched with either $w_1$ or $a$. Either choice forces the matching
    on the ``chain of boxes'' above $x_1$ remain identical to $M$. But then if
    $x_1$ is matched with $a$, there are no vertices to which $w_1$ can be
    matched. So $x_1$ must be matched with $w_1$, and the choice of edge for
    $x_2$, $y_1$, and $y_2$ is forced symmetrically, giving the matching $M$.

    Similarly, there are exactly two perfect matchings of $H$. Vertex $u$ is
    matched with either $w_1$ or $w_2$, and either choice determines all other
    edges. In particular, if $u$ is matched with $w_1$, then $x_1$ must be
    matched with $a$, and $y_1$ with $z_1$, and so on along the entire
    $12$-cycle containing $u$ and $v$. The edges on the four ``chains of
    boxes'' are then also completely determined. The other case, when $u$ is
    matched with $w_2$, is symmetric.

    We now argue that $|\mcN_H(x_1, v)| \ge 2^k$.
    Starting from the matching $M' \in \mcN_H(x_1, v)$ depicted in
    Figure~\ref{fig:gadget-alt}, each of the $k$ copies of $C_4$ in the chain
    of boxes above $x_1$ can be independently alternated, giving $2^k$ distinct
    matchings in $\mcN_H(x_1, v)$. \qed
\end{proof}

The torpid mixing gadget already suffices on its own to show that the Markov
chain $\mfX_{\jsv}$ defined in~\cite{JSV} is torpidly mixing. In
particular, the conductance out of the set $\mcN_H(x_1, v) \subseteq \Omega(H)$
is $2^{-\Omega(k)}$. In order to prove the stronger claim of
Theorem~\ref{thm:jsv-fails}, that every Markov chain of $\jsv$-type fails
to efficiently sample perfect matchings, we construct a slightly larger graph
from copies of the torpid mixing gadgets.

\begin{definition}\label{def:counterexample}
    The \textbf{counterexample graph} $G_k$ is the graph depicted in
    Figure~\ref{fig:counterexample}.  It is defined
    by replacing every third edge of the twelve-cycle $C_{12}$ with the gadget
    $H_k$ defined in Figure~\ref{fig:gadget}. Specifically, let $\{u_i, v_i\}$
    be the $3i$-th edge of $C_{12}$ for $i \in \{1, \ldots, 4\}$.
    We delete each edge $\{u_i, v_i\}$ and replace it
    with a copy of $H$, identifying the vertices $u$ and $v$ of $H$ with
    vertices $u_i$ and $v_i$ of $C_{12}$. The resulting graph is $G_k$. Thus, of
    the $12$ original vertices in $C_{12}$, $8$ of the corresponding vertices
    in $G_k$ participate in a copy of the gadget $H$, and $4$ do not. These $4$
    vertices of $G_k$ which do not participate in any copy of the gadget $H$ are
    labeled $t_1, \ldots, t_4$ in cyclic order, and the copies of the gadget
    $H$ are labeled $H_1, \ldots H_4$ in cyclic order, with $H_1$ coming
    between $t_1$ and $t_2$, and so on. Thus, $t_1$ is adjacent to $u_1$ and
    $v_4$, $t_i$ is adjacent to $u_i$ and $v_{i-1}$ for $i \in \{2,\ldots,
    4\}$, and $H_i$ contains both $u_i$ and $v_i$.
\end{definition}

\begin{figure}[h!]
    \begin{centering}
    \includegraphics{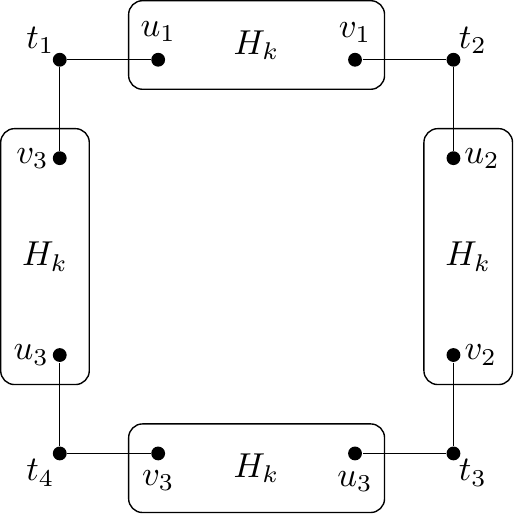}

    \end{centering}
    \caption{The ``counterexample graph'' $G_k$ on which $\jsv$ is torpidly
    mixing. The boxes labeled $H_k$ represent copies of the torpid mixing
    gadget of Definition~\ref{def:torpid-gadget}.}
    \label{fig:counterexample}
\end{figure}

In particular, $G_k$ has $4|V(H)| + 4 = 64k + 8$ vertices.

The perfect and near-perfect matchings of $G_k$ are naturally divided into four
intersecting families. For $i \in \{1,\ldots, 4\}$ we define $S_i$
to be the collection of (perfect and near-perfect) matchings $M \in
\Omega(G_k)$ such that the restriction of $M$ to $H_i$ has two holes, at $u_i$ and
$v_i$, i.e., such that the vertices $u_i$ and $v_i$ either have holes in
$M$ or are matched outside of $H_i$.

\begin{lemma}\label{lem:counterexample-matchings}
    The counterexample graph $G_k$ has exactly $8$ perfect matchings. Of these,
    $4$ are in $S_1 \cap S_3 \setminus (S_2 \cup S_4)$ and $4$ are in $S_2 \cap
    S_4 \setminus (S_1 \cup S_3)$.
\end{lemma}
\begin{proof}
    The graph $G_k$ has exactly $8$ perfect matchings. To obtain a matching in
    $S_1 \cap S_3 \setminus (S_2 \cup S_4)$, we may without loss of generality
    start by matching the vertices in $H_1$ and $H_3$ according to a matching
    in $\mcN_{H_1}(u_1, v_1)$ or $\mcN_{H_3}(u_3,
    v_3)$, respectively. We must then match $t_1$ with $u_1$, $t_2$ with $v_1$,
    $t_3$ with $u_3$, and $t_4$ with $v_3$. Finally, we must match the
    remaining vertices according to a perfect matching on each of $H_2$ and
    $H_4$. By Lemma~\ref{lem:torpid-gadget}, there are two perfect matchings on each
    of $H_2$ and $H_4$, and a unique matching in each of $\mcN_{H_1}(u_1, v_1)$
    and $\mcN_{H_3}(u_3, v_3)$, so indeed there are four matchings in $(S_1
    \cap S_3) \setminus (S_2 \cup S_4)$. Similarly, there are exactly four
    matchings in $(S_2 \cap S_4) \setminus (S_1 \cup S_3)$.

    To see that there are no other perfect matchings, let $M$ be an arbitrary
    perfect matching of $G_k$. Then $t_1$ is matched either with $u_1$ or
    $v_4$. Suppose $t_1$ is matched with $u_1$. Then $v_4$ is matched within
    $H_4$. Since $H_4$ has an even number of vertices, $u_4$ must also be
    matched within $H_4$, and hence $M$ induces a perfect matching on $H_4$.
    Continuing in a similar fashion, $M$ must also induce a perfect matching on
    $H_2$. Then the restriction of $M$ to $H_1$ or $H_3$ has holes at $u_1$ and
    $v_1$, and at $u_3$ and $v_3$, respectively, so $M \in S_1 \cap S_3
    \setminus (S_2 \cup S_4)$. Symmetrically, if $t_1$ is matched with $v_4$
    then $M \in S_2 \cap S_4 \setminus (S_1 \cup S_3)$. \qed
\end{proof}

In the proof below, we use the notation $\mcN(M)$ denote the collection of
matchings with the hole pattern as $M$. That is, $\mcN(M) = \mcP$ if $M \in
\mcP$, and $\mcN(M) = \mcN(u,v)$ if $M \in \mcN(u,v)$.

\begin{proof}[Proof of Theorem~\ref{thm:jsv-fails}]
    Let $G_k$ be the counterexample graph of
    Definition~\ref{def:counterexample}.  We will show that the set $S_1 \cup
    S_3 \subseteq \Omega(G_k)$ has poor conductance, unless the stationary
    probability of $\mcP_{G_k}$ is small.  We will write $A = S_1 \cup S_3$ and
    $\overline{A} = \Omega(G_k)\setminus (S_1\cup S_3)$.

    Let $M \in A$ and $M' \in \overline{A}$ be such that $P(M, M') > 0$. We
    claim that neither $M$ nor $M'$ are perfect matchings.  Assume without loss
    of generality that $M \in S_1$.  If $M \in S_1$ is a perfect matching, then
    $M \in P_2$ and so $M \in S_3$. The only legal transitions from $M$ to
    $\Omega \setminus S_1$ are those that introduce additional holes within
    $H_1$, but none of these transitions to a matching outside of $S_3$. Hence,
    $M$ cannot be perfect.  But if $M'$ is perfect, then $M' \in P_1$, and so
    $M'$ induces a perfect matching on $S_1$. But then the transition from $M$
    to $M'$ must simultaneously affect $u_1$ and $v_1$, and no such transition
    exists.

    We denote by $\partial \overline{A}$ the set of matchings $M' \in
    \overline{A}$ such that there exists a matching $M \in A$ with $P(M, M') >
    0$. We claim that for every matching $M' \in \overline{A}$, we have
    \begin{equation}\label{eq:small-boundary}
        |\mcN(M') \cap \partial \overline{A}| \le 2^{k-1} |\mcN(M')|\,.
    \end{equation}
    Let $M' \in \partial\overline{A}$, and let $M \in A$ be such that
    $P(M, M') > 0$.  Suppose first that $M \in
    S_1$.  Label the vertices of $H_1$ as in Figure~\ref{fig:gadget},
    identifying $u_1$ with $u$ and $v_1$ with $v$.  Let $N$ be the matching on
    $H = H_1$ induced by $M$, and let $N'$ be the matching on $H_1$ induced by
    $M'$. We have $N \in \mcN_{H}(u_1, v_1)$.  But by
    Lemma~\ref{lem:torpid-gadget}, we have $|\mcN_H(u_1, v_1)| = 1$, i.e., $N$
    is exactly the matching depicted in Figure~\ref{fig:gadget}. The only
    transitions that remove the hole at $u$ are the two that shift the hole to
    $x_1$ or $x_2$, and the only transitions that remove the hole at $v$ are
    the two that shift the hole to $y_1$ or $y_2$.  So, without loss of
    generality, by the symmetry of $G_k$, we have $N' \in \mcN_H(x_1, v_1)$.
    By Lemma~\ref{lem:torpid-gadget}, $|\mcN_H(x_1, v_1)| \ge 2^k$, but only
    one matching in $\mcN_H(x_1, v_1)$ has a legal transition to $N$.
    Therefore, if we replace the restriction of $M'$ to
    $H_1$ with any other matching in $\mcN_H(x_1, v_1)$, we obtain another
    matching $M'' \in \mcN(M')$, but $M''$ has no legal transition to any
    matching in $\mcN(M)$. Hence, only a $2^{-k}$-fraction of $\mcN(M')$ has a
    legal transition to $S_1$, and similarly only a $2^{-k}$-fraction of
    $\mcN(M')$ has a legal transition to $S_3$. In particular, we have proved
    Eq.~\eqref{eq:small-boundary}.

    From Eq.~\eqref{eq:small-boundary}, it immediately follows that the
    stationary probability of $\partial\overline{A}$ is
    \begin{equation}\label{eq:small-boundary2}
        \pi(\partial\overline{A}) = \sum_{M' \in \partial\overline{A}} \pi(M')
        = \sum_{M' \in \overline{A}} \pi(M')\frac{|\mcN(M')\cap
        \partial\overline{A}|}{|\mcN(M')|} = 2^{-k+1}\pi(\overline{A})
    \end{equation}

    We now compute
    \begin{align*}
        \sum_{\substack{M \in A,  M' \in
        \overline{A} \\ P(M, M') > 0}}\pi(M)P(M, M')
        = \sum_{\substack{M \in A, M' \in
        \overline{A} \\ P(M, M') > 0}}\pi(M')P(M', M)
        &\le \pi(\partial(\overline{A})) \\ &< 2^{-k+1}\pi(\overline{A}),
    \end{align*}
    where we first use the detailed balance condition and then
    Eq.~\eqref{eq:small-boundary2}.

    Now by \eqref{eq:conductance} and the
    definition of conductance, we have
    \[
        \frac{1}{4 \tau_{\mfX}} < \Phi(A)
        < 2^{-k}\frac{\pi(\overline{A})}{\pi(A)}\,.
    \]
    In particular, if $\tau_{\mfX} < 2^{k/2 - 2}$, then $\pi(\overline{A}) >
    2^{k/2+1}\pi(A)$. Suppose this is the case. By
    Lemma~\ref{lem:counterexample-matchings}, half of the perfect matchings of
    $G_k$ belong to $A$. In particular, $\pi(\mcP_{G_k}) \le 2\pi(A) <
    2^{-k/2+2}$. Hence, either the stationary probability of $\mcP$ is at most
    $2^{-k/2+2} = \exp(-\Omega(n))$, or the mixing time of $\mfX$ is at least
    $2^{k/2-2} = \exp(\Omega(n))$. \qed
\end{proof}

We remark that the earlier Markov chain studied by Broder~\cite{Broder} and Jerrum and
Sinclair~\cite{JS} is also torpidly mixing on the counterexample graph of
Definition~\ref{def:counterexample}, since the ratio of near-perfect matchings to
perfect matchings is exponential~\cite{JS}.

\section{Chains Based on Edmonds' Algorithm}
\label{sec:blossoms-fail}

Given that Edmonds' classical algorithm for \emph{finding} a perfect matching
in a bipartite graph requires the careful consideration of odd cycles in the
graph, it is reasonable to ask whether a Markov chain for counting perfect
matchings should also somehow track odd cycles. In this section, we briefly
outline some of the difficulties of such an approach.

A \emph{blossom} of length $k$ in a graph $G$ equipped with a matching $M$ is
simply an odd cycle of length $2k + 1$ in which $k$ of the edges belong to $M$.
Edmonds' algorithm finds augmenting paths in a graph by exploring the
alternating tree rooted at an unmatched vertex, and contracting blossoms to a
vertex as they are encountered. Given a blossom $B$ containing an unmatched
vertex $u$, there is an alternating path of even length to every vertex $v \in
B$. \emph{Rotating} $B$ to $v$ means shifting the hole at $u$ to $v$ by
alternating the $u$-$v$ path in $B$.

Adding rotation moves to a Markov chain in the style of $\jsv$ is an attractive
possible solution to the obstacles presented in the previous section. Indeed,
if it were possible to rotate the $7$-cycle containing $u$ and $a$ in the graph in
Figure~\ref{fig:gadget}, it might be possible to completely avoid problematic
holes at $x_1$ or $x_2$.

The difficulty in introducing such an additional move the Markov chain $\jsv$
is in defining the set of \emph{feasible} blossoms that may be rotated, along with
a probability distribution over such blossoms. In order to be useful, we must
be able to \emph{efficiently sample} from the feasible blossoms at a given
near-perfect matching $M$. Furthermore, the feasible blossoms must respect time
reversibility: if $B$ is feasible when the hole is at $u \in B$, then it must
also be feasible after rotating the hole to $v \in B$; reversibility of the Markov chain is needed
so that we understand its stationary distribution.  Finally, the feasible
blossoms must be rich enough to avoid the obstacles outlined in the previous
section.

The set of ``minimum length'' blossoms at a given hole vertex $u$ satisfies the
first criterion of having an efficient sampling algorithm. But it is easy to
see that if only minimum length blossoms are feasible, then the obstacles
outlined in the previous section will still apply (simply by adding a $3$-cycle
at every vertex). Moreover, families blossoms characterized by minimality may
struggle to satisfy the second criterion of time-reversibility. In
Figure~\ref{fig:blossom}, there is a unique blossom containing $u$, but after
rotating the hole to $v$, it is no longer minimal.

\begin{figure}
    \begin{centering}
    \includegraphics{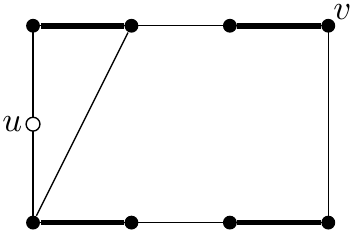}

    \end{centering}
    \caption{After rotating the blossom so that the hole is moved from $u$ to
    $v$, the blossom is no longer ``minimal''.}
    \label{fig:blossom}
\end{figure}

On the other hand, the necessity of having an efficient sampling algorithm for
the feasible blossoms already rules out the simplest possibility, namely, the
uniform distribution over \emph{all} blossoms containing a given hole vertex
$u$.  Indeed, if we could efficiently sample from the uniform distribution over
all blossoms containing a given vertex $u$, then by an entropy argument we
could find arbitrarily large odd cycles in the graph, which is NP-hard.

\begin{theorem}
\label{thm:all-blossoms-hard}
Let {\sc Sampling Blossoms} problem be defined as follows. The
input is an undirected graph $G$ and a near-perfect matching $M$
with holes at $w,r\in V(G)$. The output is a uniform sample from
the uniform distribution of blossoms containing $w$. Unless NP=RP
there is no randomized polynomial-time sampler for {\sc Sampling
Blossoms}.
\end{theorem}

\begin{proof}
We reduce from the problem of finding the longest $s$-$t$-path in
a directed graph $H$ (ND29 in~\cite{GJ}). We construct an instance
of {\sc Sampling Blossoms}, that is, $G$ and $M$ as follows. For
every $v\in V(H)$ we add two vertices $v_0,v_1$ into $V(G)$ and
also add $\{v_0,v_1\}$ into $M$. For every edge $(u,v)\in E(H)$ we
add edge $\{u_1,v_0\}$ into $E(G)$. Finally we add $w,r$ into
$V(H)$ and $\{w,s_0\},\{t_1,w\}$ into $E(H)$.

Note that there is one-to-one correspondence between blossoms that
contain $w$ in $G$ and $s$-$t$-paths in $H$. We now modify $G$ to
``encourage longer paths''. We replace each $\{v_0,v_1\}$ edge in
$G$ by a chain of boxes (with $\ell$ boxes) and replace
$\{v_0,v_1\}$ in $M$ by the unique perfect matching of the chain
of boxes. In the modified graph $G$ for every $s$-$t$-path $p$ in
$H$ there are now $2^{k \ell}$ blossoms that contain $w$ in $G$,
where $k$ is the number of vertices in $p$.

Taking $\ell = n^2$ a uniformly random blossom that contains $w$
in $G$ will with probability $1-o(1)$ correspond to a longest
$s$-$t$-path in $H$ (the number of $s$-$t$-paths is bounded by
$(n+1)^n = 2^{O(n\log n)}$ and hence the fraction of blossoms
corresponding to non-longest $s$-$t$-paths is $2^{O(n\log n)}
2^{-n^2} = o(1)$). \qed
\end{proof}

\section{A Recursive Algorithm}
\label{sec:EG}

We now explore a new recursive algorithm for counting matchings, based on the
Gallai--Edmonds decomposition. In the worst case, this algorithm may still
require exponential time. However, for graphs that have additional structural
properties, for example, those that are ``sufficiently close to bipartite'' in
a sense that will be made precise, our recursive algorithm runs in polynomial
time. In particular, it will run efficiently on examples similar to those used
to prove torpid mixing of Markov chains in the previous section.

We now state the algorithm.  It requires as a subroutine an algorithm for
computing the permanent of the bipartite adjacency matrix of a bipartite graph
$G$ up to accuracy $\eps$. We denote this subroutine by $\permanentalg(G,
\eps)$.  The $\permanentalg$ subroutine requires time polynomial in $|V(G)|$
and $1/\eps$ using the algorithm of Jerrum, Sinclair, and
Vigoda~\cite{JSV}, but we use it as a black-box.

\begin{algorithm}
    \caption{Recursive algorithm for approximately counting the number of perfect matchings in
    a graph}\label{alg:recursive}
    \begin{algorithmic}[1]
        \Procedure{Recursive-Count}{$G,\eps$}
            \State If $V(G) = \emptyset$, return $1$.
            \State Choose $u \in V(G)$.
            \State Compute the Gallai--Edmonds decomposition of $G-u$.
            \ForAll{$v \in D(G-u)$}
                \State $H_v \gets$ the connected component of $G-u$ containing $v$
                \State $m_v \gets$ \Call{Recursive-Count}{$H_v - v, \eps/(2n)$}
            \EndFor
            \State $m_C \gets$ \Call{Recursive-Count}{$C(G-u), \eps/3$}
            \State Let $X = A(G-u) \cup \{u\}$, and let $Y$ be the set of connected
            components in $D(G-u)$. Let $G'$ be the bipartite graph on $(X,Y)$
            defined as follows: for every $x \in X$ and $H \in Y$, if $x$ has any
            neighbors in $H$ in $G'$, add an edge $\{x, H\}$ in $G'$ with weight
        \[
            w(x, H) = \sum_{v \in N(x) \cap H} m_v\,.
        \]
        \State \textbf{return} $m_C * \permanentalg(G', \eps/3)$
        \EndProcedure
    \end{algorithmic}
\end{algorithm}

We first argue the correctness of the algorithm.

\begin{theorem}\label{thm:main}
    Algorithm~\ref{alg:recursive} computes the number of perfect matchings in $G$ to
    within accuracy $\eps$.
\end{theorem}
\begin{proof}
    We show that the algorithm is correct for graphs on $n$ vertices, assuming
    it is correct for all graphs on at most $n-1$ vertices.

    We claim that permanent of the incidence matrix of $G'$ defined on line 10
    equals the number of perfect matchings in $G$. Indeed, every perfect
    matching $M$ of $G$ induces a maximum matching $M_u$ on $G-u$. By the
    Gallai--Edmonds theorem, $M_u$ matches each element of $A(G')$ with a
    vertex from a distinct component of $D(G')$, leaving one component of
    $D(G')$ unmatched.  Vertex $u$ must therefore be matched in $M$ with a
    vertex from the remaining component of $D(G')$.  Therefore, $M$ induces a
    perfect matching $M'$ on $G'$. Now let $H_x \in Y$ be the vertex of $G'$
    matched to $x$ for each $x \in X$. Then the number of distinct matchings of
    $G$ inducing the same matching $M''$ on $G''$ is exactly
    \[
        \prod_{x \in X}\sum_{v \in N(x) \cap H_x} m_v = \prod_{x \in X} w(x,
        H_x)
    \]
    which is the contribution of $M'$ to the permanent of $G'$. Similarly,
    from an arbitrary matching $M'$ of $G'$, with $H_x$ defined as above, we
    obtain $\prod_{x \in X} w(x, H_x)$ matchings of $G$, proving the claim.

    Hence, it suffices to to compute the permanent of the incidence matrix of
    $G'$ up to accuracy $\eps$.  We know the entries of the incidence matrix
    up to accuracy $\eps/(2n)$, and $(1 + \eps/(2n))^{n/2} < 1 + \eps/3$ for $\eps$
    sufficiently small.  Therefore, it suffices to compute the permanent of our
    approximation of the incidence matrix up to accuracy $\eps/3$ to get
    overall accuracy better than $\eps$. \qed
\end{proof}

The running time of Algorithm~\ref{alg:recursive} is sensitive to the choice of
vertex $u$ on line 3. If $u$ can be chosen so that each component of
$D(G-u)$ is small, then the algorithm is an efficient divide-and-conquer
strategy. More generally, if $u$ can be chosen so that each component of
$D(G-u)$ is in some sense ``tractable'', then an efficient divide-and conquer
strategy results. In particular, since it is possible to exactly count the
number of perfect matchings in a factor-critical graph of bounded order in
polynomial time, we obtain an efficient algorithm for approximately counting
matchings in graphs whose factor-critical subgraphs have bounded order. This is
the sense in which Algorithm~\ref{alg:recursive} is efficient for graphs
``sufficiently close'' to bipartite.

\begin{theorem}\label{thm:fpt-alg}
    Suppose every factor-critical subgraph of $G$ has order at most $k$. Then
    the number of perfect matchings in $G$ can be counted to within accuracy
    $\eps$ in time $2^{O(k)} \poly(n,1/\eps)$.
\end{theorem}

The essential idea of the proof is to first observe that a factor-critical
graph can be shrunk to a graph with $O(k)$ edges having the same number of
perfect matchings after deleting any vertex. The number of perfect matchings
can then be counted by brute force in time $2^{O(k)}\poly(n)$. This procedure
replaces the recursive calls on line $6$ of the algorithm.

\begin{proof}
We first observe that if $H$ is a factor-critical graph of order $k$ with
$n$ vertices, then the number of perfect matchings in $H-v$ can be counted
exactly in time $2^{O(k)} \poly(n)$ for every vertex $v$. Writing $d_u$ for the degree of a
vertex $u$, we have
    \begin{equation}\label{eq:fc-degree}
    \sum_{u \in V(H)} (d_u - 2) = 2(k-1),
    \end{equation}
    since adding
one ear to a graph adds some number of vertices of degree $2$, and
increases the degree of two existing vertices by one each, or one vertex by
two. Fix $v \in H$, and suppose there is a vertex $u$ of degree $2$ in $H - v$,
with neighbors $w_1$ and
$w_2$. Let $H'$ denote the multigraph obtained from $H - v$ by contracting the edges from
$u$ to $w_1$ and $w_2$, so $H'$ has two fewer
vertices than $H$, and has a vertex $w$ with the same multiset of neighbors as
$w_1$ and $w_2$ (excluding $v$). Then there is a bijection between the perfect
matchings of $H'$ and of $H-v$; each perfect matching of $H'$ lifts to a
matching of $H-v$ with a hole at $u$ and exactly one of $w_1$ or $w_2$, and
each perfect matching of $H-v$ projects to a perfect matching of $H'$ by
ignoring the matched edge at $u$. Hence, we may contract away all degree-$2$
vertices of $H-v$, and obtain a graph with the same number of perfect matchings in which
every vertex (save at most two of degree $1$, the former neighbors of $v$)
has degree at least $3$. Then since the contraction does not change the sum in
    Eq.~\eqref{eq:fc-degree}, we have
    \[
        3(|V(H')|-2) \le \sum_{u \in V(H')} d_u \le 2(k-1) + 2|V(H')|
        \]
    and hence $H'$ has $O(k)$ edges, and the perfect matchings of $H'$ can
    be enumerated in time $2^{O(k)}$.

    Now we modify Algorithm~\ref{alg:recursive} to run in time $2^{O(k)}
    \poly(n, 1/\eps)$. First, we delete all edges not appearing in any perfect
    matching, call \Call{Recursive-Count}{$G_i,\eps/(2n)$} on each
    connected component $G_i$, and multiply the results of all of these calls
    to estimate the number of perfect matchings in $G$. We have $C(G_i-u) =
    \emptyset$ for each such component $G_i$ and every vertex $u \in V(G_i)$,
    since edges leaving $C(G_i - u)$ cannot appear in any matching of $G_i - u$.
    Therefore, the recursive call on line 8 of the algorithm can be eliminated.
    On line 6, instead of computing $m_v$ by a recursive
    call, we instead use the procedure described above to compute it in time
    $2^{O(k)}$. Hence, Algorithm~\ref{alg:recursive} requires $O(n)$ calls to a
    procedure that takes time $2^{O(k)}$. The other lines of
    Algorithm~\ref{alg:recursive} require only polynomial time in $n$ and
    $1/\eps$, so in all Algorithm~\ref{alg:recursive} requires time
    $2^{O(k)}\poly(n, 1/\eps)$. \qed
\end{proof}

We note that Theorem~\ref{thm:fpt-alg} is proved by eliminating recursive calls in
the algorithm. Although the recursive calls of Algorithm~\ref{alg:recursive}
can be difficult to analyze, they can also be useful, as we now demonstrate by
showing that Algorithm~\ref{alg:recursive} runs as-is in polynomial time on the
counterexample graph of Definition~\ref{def:counterexample}, for appropriate choice of
the vertex $u$ on the line 3 of the algorithm.

\begin{theorem}\label{thm:counterexample-alg}
    Algorithm~\ref{alg:recursive} runs in polynomial time on the counterexample
    graph of Definition~\ref{def:counterexample}, for appropriate choice of the vertex
    $u$ on the line 3 of the algorithm
\end{theorem}
\begin{proof}
    After deleting the vertices $u$ and $v$ from the torpid mixing gadget $H$ in
    Figure~\ref{fig:gadget}, no odd cycles remain the graph $H$. Let $U$ denote
    the set of all four copies of the vertices $u$ and $v$ appearing in the
    counterexample graph $G$, so $|U| = 8$. With every recursive call
    \Call{Recursive-Count}{$G', \eps'$}, if $U \cap V(G') \ne \emptyset$, we
    choose $u \in U\cap V(G')$. Hence, after $8$ recursive calls, there are no
    odd cycles remaining in $G'$, and each factor-critical subgraph is a single
    vertex. When $U \cap V(G') = \emptyset$, we choose $u$ so that $A(G'-u) =
    \Omega(n)$---for example taking $u$ at one end of a chain of boxes---so
    that the overall recursive depth is $O(1)$. \qed
\end{proof}

\subsubsection*{Acknowledgements}
This research was supported in part by NSF grants CCF-1617306, CCF-1563838,
CCF-1318374, and CCF-1717349. The authors are grateful to Santosh Vempala for
many illuminating conversations about Markov chains and the structure of
factor-critical graphs.

\end{document}